\documentclass[a4paper,UKenglish,cleveref, autoref, thm-restate]{lipics-v2021}

\hideLIPIcs  

\title{Hexasort -- The Complexity of Stacking Colors on Graphs} 

\titlerunning{Hexasort -- The Complexity of Stacking Colors on Graphs} 

\author{Linus Klocker}{TU Wien, Vienna, Austria}{linus.klocker@tuwien.ac.at}{}{}

\author{{Simon Dominik} Fink}{Department of Algorithms and Complexity, TU Wien, Vienna, Austria \and \url{https://informatics.tuwien.ac.at/people/simon-fink}}{simon.fink@tuwien.ac.at}{https://orcid.org/0000-0002-2754-1195}{}

\authorrunning{L. Klocker and S.\,D. Fink} 

\Copyright{Linus Klocker and Simon Dominik Fink} 

\ccsdesc[100]{Theory of computation \textrightarrow \, Problems, reductions and completeness} 

\keywords{Hexasort, offline color stacking on graphs, NP-complete, polynomial-time solvable, dynamic programming} 

\relatedversion{} 

\funding{Vienna Science and Technology Fund (WWTF) grant [10.47379/ICT22029] and Austrian Science Fund (FWF) grant [10.55776/Y1329]}

\nolinenumbers 

\usepackage{etoolbox}
\newtoggle{long}\toggletrue{long} 

\iftoggle{long}{
  \NewDocumentEnvironment{prooflater}{m}{\begin{proof}}{\end{proof}}
  \NewDocumentEnvironment{proofsketch}{o +b}{}{}
  \newcommand{\restateref}[1]{}

  \NewDocumentEnvironment{statelater}{m}{}{}
  \NewDocumentCommand{\onlyShort}{+m}{}
  \NewDocumentCommand{\onlyLong}{+m}{#1}
}{
  \NewDocumentEnvironment{prooflater}{m +b}{ %
    \newcounter{#1-usages}\setcounter{#1-usages}{1}%
    \AtEndDocument{\ifnumequal{\value{#1-usages}}{0}{\todo[inline]{use prooflater `#1'}}{}}%
    \expandafter\global\expandafter\def\csname#1\endcsname{\stepcounter{#1-usages}\begin{proof}#2\end{proof}}%
  }{}
  \NewDocumentEnvironment{proofsketch}{O{Proof sketch.}}{\begin{proof}[#1]}{\end{proof}}
  \usepackage{apptools}
  \newcommand{\restateref}[1]{[\IfAppendix{\hyperref[#1]{$\star$}}{\hyperref[#1*]{$\star$}}]}

  \NewDocumentEnvironment{statelater}{m +b}{%
    \newcounter{#1-usages}\setcounter{#1-usages}{0}%
    \AtEndDocument{\ifnumequal{\value{#1-usages}}{0}{\todo[inline]{use statelater `#1'}}{}}%
    \expandafter\global\expandafter\def\csname#1\endcsname{\stepcounter{#1-usages}#2}%
  }{\ignorespacesafterend}
  \NewDocumentCommand{\onlyShort}{+m}{#1}
  \NewDocumentCommand{\onlyLong}{+m}{}
}

\let\oldrestatable\restatable
\def\restatable{\expandafter\oldrestatable}

\makeatletter
\pretocmd{\thmt@rst@storecounters}{\Hy@SaveLastskip}{}{}
\apptocmd{\thmt@rst@storecounters}{\Hy@RestoreLastskip}{}{}
\makeatother

\EventEditors{John Iacono}
\EventNoEds{1}
\EventLongTitle{13th International Conference on Fun with Algorithms (FUN 2026)}
\EventShortTitle{FUN 2026}
\EventAcronym{FUN}
\EventYear{2026}
\EventDate{May 18--22, 2026}
\EventLocation{Porquerolles, France}
\EventLogo{}
\SeriesVolume{366}
\ArticleNo{16}

\usepackage[export]{adjustbox}
\usepackage{todonotes}
\usepackage{xspace}

\usepackage{cite} 
\usepackage{fnpct} 
\usepackage{microtype} 

\newcommand{\hs}{\textsc{Hexasort}\xspace}
\newcommand{\fhs}{\textsc{Fitting Hexasort}\xspace}
\newcommand{\ehs}{\textsc{Empty Hexasort}\xspace}

\definecolor{BrickRed}{rgb}{0.72, 0.08, 0.04}
\definecolor{BlueViolet}{rgb}{0.14, 0.09, 0.96}
\definecolor{ForestGreen}{rgb}{0.08, 0.88, 0.11}

\renewcommand{\c}{\ensuremath{x}\xspace}

\newcommand{\cblack}{\ensuremath{c_\text{black}}\xspace}
\newcommand{\cblue}{\ensuremath{c_\text{\textcolor{BlueViolet}{blue}}}\xspace}
\newcommand{\cred}{\ensuremath{c_\text{\textcolor{BrickRed}{red}}}\xspace}
\newcommand{\cgreen}{\ensuremath{c_\text{\textcolor{ForestGreen}{green}}}\xspace}

\newcommand{\inlineRedStack}{\raisebox{-0.15em}{\includegraphics[height=0.9em, page=12]{graphics/partition/partition_spider.pdf}}\xspace}
\newcommand{\inlineBlueStack}{\raisebox{-0.15em}{\includegraphics[height=0.9em, page=13]{graphics/partition/partition_spider.pdf}}\xspace}
\newcommand{\inlineGreenStack}{\raisebox{-0.15em}{\includegraphics[height=0.9em, page=14]{graphics/partition/partition_spider.pdf}}\xspace}
\newcommand{\inlineBlackStack}{\raisebox{-0.15em}{\includegraphics[height=0.9em, page=15]{graphics/partition/partition_spider.pdf}}\xspace}

\begin{document}

\maketitle


\begin{abstract}
Many popular puzzle and matching games have been analyzed through the lens of computational complexity. Prominent examples include Sudoku~\cite{DBLP:journals/ieicet/YatoS03}, Candy Crush~\cite{DBLP:conf/cig/GualaLN14}, and Flood-It~\cite{DBLP:conf/tamc/FellowsRDS17}.
A common theme among these widely played games is that their generalized decision versions are NP-hard, which is often thought of as a source of their inherent difficulty and addictive appeal to human players. 
In this paper, we study a popular single-player stacking game commonly known as \textsc{Hexasort}.
The game can be modelled as placing colored stacks onto the vertices of a graph, where adjacent stacks of the same color merge and vanish according to deterministic rules.
We prove that \textsc{Hexasort} is NP-hard, even when restricted to single-color stacks and progressively more constrained classes of graphs, culminating in strong NP-hardness on trees of either bounded height or degree.
Towards fixed-parameter tractable algorithms, we identify settings in which the problem becomes polynomial-time solvable and present dynamic programming algorithms.
\end{abstract}

\section{Introduction}
We study the computational complexity of the single-player stacking game Hexasort.
There exist many different browser-based or mobile phone app implementations\footnote{See e.g.\ \href{https://hexasort.io/}{hexasort.io}, \href{https://games-by-sam.com/bee-sort-by-sam/}{games-by-sam.com/bee-sort-by-sam}, or \href{https://www.spielaffe.de/Spiel/Hexa-Sort-3D}{spielaffe.de/Spiel/Hexa-Sort-3D}.}, often with various variations to the mechanics.
Common to all is that they are played on a hexagonal game board, where a given sequence of stacks of colored hexagons need to be placed onto the empty cells.
Whenever the topmost hexagons of adjacent stacks have the same color, all hexagons of said color move from one stack onto the other; see \Cref{fig:hexasort-mobile-game}.
If this results in sufficiently many same-colored hexagons on top of each other, all of them vanish.
The goal of the game is to place the full input sequence without running out of empty cells.

\begin{figure}[t]
    \centering
    \begin{subfigure}[b]{0.3\columnwidth}
        \centering
        \includegraphics[width=.8\textwidth]{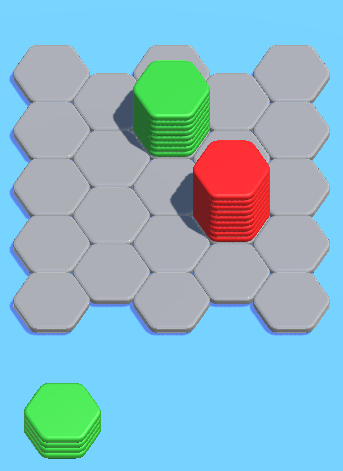}
        \caption{Original game screenshot}
        \label{fig:hexasort-1}
    \end{subfigure}
    \begin{subfigure}[b]{0.3\columnwidth}
        \centering
        \includegraphics[width=.8\textwidth,page=1]{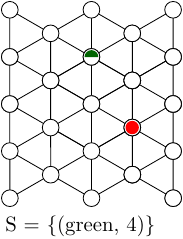}
        \caption{Graph abstraction}
        \label{fig:hexasort-2}
    \end{subfigure}
    \begin{subfigure}[b]{0.3\columnwidth}
        \centering
        \includegraphics[width=.8\textwidth,page=2]{graphics/hexasort-mobile-game/intro-figure-configuration.pdf}
        \caption{Result after placing \inlineGreenStack}
        \label{fig:hexasort-3}
    \end{subfigure}
    \caption{
        Example of a \textsc{Fitting Hexasort} game with its graph representation (not showing stack heights). After placing \inlineGreenStack, the green stacks merge, reach $t$ and vanish. 
    }
    \label{fig:hexasort-mobile-game}
\end{figure}

In our model of the game, we generalize the game board to an arbitrary graph, where the colored stacks are placed onto its vertices.
We restrict our consideration to monochromatic stacks and define that whenever a stack is placed adjacent to stacks of the same color, the old stacks completely move onto the new one.
We will not consider any further complications to the mechanics such as initially non-empty game boards and blockers that can only be removed after having performed a certain amount of merges, although we will later show that similar mechanics can be achieved through gadgets.
An instance $(G,S,t)$ of \textsc{Hexasort} thus consists of a graph $G = (V,E)$, a \emph{threshold} $t \in \mathbb{N^+}$, and a sequence $S$ of \emph{stacks} having colors $C$, where each stack $s=(c,h)\in S$ has a \emph{color} $c\in C$ and a \emph{height} $h \in \mathbb{N^+}$.

A \emph{configuration} of the game assigns to each vertex either a stack or nothing, that is, it is a function $\delta : V(G) \to C \times \mathbb{N}^{+} \;\cup\; \perp$ where $\delta(v) = \perp$ if vertex $v$ is currently \emph{empty}, and otherwise $\delta(v) =(c,h)$ if $v$ hosts a stack of color $c$ and height $h$.
Initially, we have $\delta(v)=\perp$ for all $v\in V$.
To place the next stack $s_i = (c,h)$ from $S$, the player chooses an empty vertex $x$.
Let $N_c[x]$ denote the non-empty neighbors of $x$ that have color $c$ and let $h'$ be the sum of their heights.
We obtain a new configuration $\delta'$ from $\delta$ by setting $\delta'(u)=\perp$ for all $u\in N_c[x]$ and updating $\delta'(x)$ depending on the accumulated height as follows.
If $h+h'\geq t$, we set $\delta'(x)=\perp$, otherwise we set $\delta'(x)=(c,h+h')$.

The central algorithmic question of the problem \ehs is whether all stacks in $S$ can be placed on the graph in their given order such that all vertices are empty again after $S$ is exhausted.
The problem variant \fhs relaxes this to only placing all stacks in $S$ in order, without the requirement for all vertices to be empty afterwards.
Note that any instance of \ehs can be turned into an equivalent one of the \textsc{Fitting} variant by appending $|V|$ stacks to $S$, all having new and distinct colors as well as a height smaller than $t$.
Conversely, any yes-instance of \ehs is trivially also a yes-instance for the \textsc{Fitting} variant.

This paper explores the boundary between tractable and intractable instances of \textsc{Fitting} and \ehs.
In \Cref{sec:np-hardness}, we show that \ehs is NP-hard, even under severe structural restrictions.
In particular, the problem is weakly NP-hard with one color on graphs consisting of two independent edges, and strongly NP-hard on trees of either bounded height or degree.
Towards fixed-parameter tractable algorithms, in \Cref{sec:dynamic-programming} we identify settings in which both variants become polynomial-time solvable and present a dynamic programming approach with running time exponential in the size of the input graph.

\subsection{Related Work}
Several games with mechanics modeling partial aspects of \textsc{Hexasort}, such as placing objects on graphs, monochromatic components, as well as merging and elimination thresholds, have been investigated.
To the best of our knowledge, the combination of these mechanics and especially the computational complexity of \textsc{Hexasort} has not been studied previously.
As we will see, the combination of non-monotone development of game states, threshold-based elimination that discards any excess, together with having to eliminate all stacks in \ehs poses an especially interesting challenge.

A recent example of a game involving the placement of objects on a graph is \textsc{Matching Match Puzzle}, analyzed by Iburi et~al.~\cite{DBLP:conf/fun/IburiU24}. In this puzzle, the player is given a colored graph and a set of matchsticks, each having a color assigned to both of its endpoints. The task is to place the matchsticks onto the edges of the graph such that, at every vertex, all incident matchstick endpoints agree on a single color consistent with the given graph coloring. The problem is NP-hard even on spiders and, interestingly, this result is established via a reduction from \textsc{3-Partition}, similar to the reduction used in \Cref{sec:strong-np}.

One well-studied example of a game relying on monochromatic components is the game \textsc{Flood-It}.
A survey by Fellows et~al.~\cite{DBLP:conf/tamc/FellowsRDS17} summarizes complexity results for two variants known as \textsc{Flood-It} and \textsc{Free-Flood-It}.
Both are played on graphs with colored vertices.
In \textsc{Flood-It}, a fixed pivot vertex is repeatedly recolored, causing its whole monochromatic connected component to change color as well.
In \textsc{Free Flood-It}, the vertex whose monochromatic component shall be recolored can be freely chosen.
In both cases, the objective is to make the entire board monochromatic using as few moves as possible.
\textsc{Flood-It} is already NP-hard on boards of size $3 \times n$ with four colors.

A game that combines merging and vanishing mechanics is \textsc{Clickomania}.
In \textsc{Clickomania}, the player is presented a grid of colored blocks.
A move consists of selecting a monochromatic connected component of at least two blocks, which removes them from the board. Blocks above the removed blocks fall down, and empty columns vanish.
The objective is to clear the board completely or minimize the number of remaining blocks.
\textsc{Clickomania} remains NP-hard even when restricted to only two colors and two columns \cite{clickomania-two-colors-two-columns}.

Outside the algorithmic study of games, the merging and threshold-based elimination of stacks in \hs exhibits parallels with the behavior of physical water droplets on digital microfluidics platforms like OpenDrop~\cite{bioengineering4020045}. 
Here, a grid of electrodes is used to manipulate discrete fluid droplets.
By activating specific electrodes, droplets can be moved, mixed, and merged on the surface.
Due to surface tension and physical constraints, a single electrode tile can only support a specific volume of liquid; exceeding this volume makes the droplet unstable.
In this analogy, stack height corresponds to droplet volume, and merging occurs based on proximity.
In addition to medical and chemical applications, the device has also been gamified using colored droplets~\cite{OpenDropYouTubeDemo}.


\subsection{Preliminaries}
We distinguish between sets (unordered collections, denoted by $\{ \dots \}$) and sequences (ordered lists, denoted by $\langle \dots \rangle$).
For two sequences $A = \langle a_1, \dots, a_n \rangle$ and $B = \langle b_1, \dots, b_m \rangle$, we define their \emph{concatenation}, denoted as $A + B$, as the new sequence formed by appending the elements of $B$ to the end of $A$, so $A + B = \langle a_1, \dots, a_n, b_1, \dots, b_m \rangle$.

We use the term \emph{three-merge} to denote a move that places the current stack such that it has exactly two adjacent non-empty vertices of the same color.
For a color $c$, we denote \emph{$sum(c, S)$} as the sum of heights of color $c$ currently remaining in $S$.

We will assume that $h \leq t$ for all $(c,h)\in S$, as any stacks exceeding $t$ can equivalently be capped to have height $t$.
Observe that no two stacks that are adjacent on $G$ can have the same color in a configuration $\delta$.
Finally, in \ehs, we assume that for any $c\in C$, we have $\mathrm{sum}(c,S) \geq t$, as the instance is trivially negative otherwise.

\section{NP-Hardness}\label{sec:np-hardness}
Our first proof that \ehs is NP-hard is via a reduction from the weakly NP-complete problem \textsc{Partition} to \ehs with a single color on graphs consisting of only two independent edges in \Cref{sec:weak-np}.
We afterwards extend this hardness result to connected graphs with nine vertices.
We further show strong NP-hardness via a reduction from \textsc{3-Partition} to \ehs on trees of bounded height or degree in \Cref{sec:strong-np}.
Thereby, we show hardness even when either the graph size and number of colors are small constants,
or when all numeric inputs are polynomially bounded in the input~size.

The figures in this section employ a visual encoding for colored stacks to ensure accessibility: \inlineRedStack, \inlineBlueStack, \inlineGreenStack, and \inlineBlackStack depict red, blue, green, and black stacks, respectively.
Note that the exact stack heights are usually not relevant for the figures and thus not shown.

\subsection{Reduction from Partition}\label{sec:weak-np}
In the \textsc{Partition} problem, we are given a multiset $P$ of positive integers and must decide whether $P$ can be partitioned into two disjoint subsets $P_1$ and $P_2=P\setminus P_1$ such that $\sum P_1 = \sum P_2$.
The \textsc{Partition} problem is a well-known special case of the \textsc{Subset Sum} problem and is (weakly) NP-hard~\cite{GareyJohnson1979}. 
We restrict our attention to \textsc{Partition} instances where (1) the sum of elements $\sum P$ is even, (2) no element $p \in P$ satisfies $p \geq \sum P/2$, and (3) the desired partition (if it exists) is not formed by the sum of a prefix of $P$.
Note that instances that violate one of these criteria can be trivially decided in polynomial time.

\textsc{Partition} can be reduced to \ehs as follows.   
Let $P = \{p_1, p_2, \dots, p_n\}$ be an arbitrary instance of \textsc{Partition}.
We construct a corresponding instance $I_H = (G, S, t)$ of \ehs with threshold $t = \sum P / 2$ where the graph $G=(V,E)$ has four vertices $V=\{v_1,v_2,v_3,v_4\}$ and two edges $E=\{v_1v_2, v_3v_4\}$. 
The input sequence of stacks $S$ is defined as $S = \langle(c, p_1), (c, p_2), \dots, (c, p_n)\rangle$, where $c$ is the only color.

\begin{theorem}\label{thm:np-partition-discon}
    \ehs on two independent edges and a single color is weakly NP-hard.
\end{theorem}
\begin{proof}
    In our above construction, it is easy to see that any solution to $P$ can be translated to a solution of the derived $I_H$ by alternatingly placing the stacks of the first partition on $v_1$ and $v_2$ while placing those of the second on $v_3$ and $v_4$.
    As each partition sums up to exactly $\sum P / 2=t$, the stacks vanish exactly once the last element of the partition is placed.
    
    For the converse direction, we first show that placing all stacks of $I_H$ exclusively on the endpoints of one of the two edges cannot empty the graph.
    Recall that we excluded the trivial case of a prefix of $P$ forming the desired partition.
    Thus, placing all stacks on one edge will at some point form a stack that exceeds $t$ by a value of at least 1 and consequently vanishes.
    Afterwards, the remaining stacks sum up to height at most $t-1$ and thereby cannot vanish.
    Thus, the stacks of $I_H$ need to be distributed onto the two independent edges.
    Assuming one edge receives stacks exceeding $t$ leaves less than $t$ for the other edge.
\end{proof}

Note that above approach could easily be extended to connected graphs if we allowed a given initial configuration with some already-placed differently-colored stack keeping both edges separate until the end.
Instead, we construct a gadget to adequately restrict the placement of stacks on a more involved but connected graph structure.
The goal is to argue that a given sequence of stacks can only be placed in such a way that two specific, non-adjacent paths of length two remain available for processing the \textsc{Partition} input sequence.
We adapt our above construction to obtain an instance $I_H^* = (G', S', t)$ where $G'$ is a connected \emph{spider} graph, that is a tree with only one vertex of degree greater than two.
Our specific spider graph $G'$ has nine vertices: one central vertex \c with degree five, two \emph{short arms} consisting of the vertices $s^1, s^2$, respectively, and three \emph{long arms} of length two; see \Cref{fig:spider_gadget}.
For long arm $i \in \{1,2,3\}$, let $\ell_a^i$ be the neighbor of \c while $\ell_b^i$ is the leaf.

We will now use the four colors: \cblack (for the \textsc{Partition} input in $S$), and \cred, \cblue, \cgreen (for the \emph{gadget stacks}).
Define the sequence of the initial gadget stacks as $S_{initial} =\langle (\cred, t-1), (\cred,t-1),(\cblue, t-1), (\cblue, t-1) \rangle$.
Define the reservation gadget stacks as $S_{reservation} = \langle(\cgreen,\lfloor \frac{t}{2}\rfloor + 1), (\cgreen,\lfloor \frac{t}{2}\rfloor + 1),
(\cgreen,\lceil \frac{t}{2}\rceil - 1), (\cgreen,\lceil \frac{t}{2}\rceil - 1) \rangle$.
Define the sequence of the connector gadget stacks $S_{connector}$ which frees the gadget of the initial stacks as $S_{connector} = \langle (\cred, t-1), (\cblue, t-1) \rangle$.
We construct the full input sequence as $S'= S_{initial} + S_{reservation} + S +  S_{connector}$.
The idea behind our construction is that the gadget stacks (non-\cblack) must vanish to empty the graph, but block most of the graph while the \cblack stacks of $S$ are processed, such that they need to be placed on the four endpoints of exactly two independent edges.
To formally show this, we need the following lemmas, which will also be useful for later settings.

\begin{lemma}[Forced three-merge]\label{lem:forced-three-merge}
    \label{lem:constraint-vanishing-three-merge}\label{lem:constraint-non-overlapping}
    Let $c$ be a color of an \ehs instance $I=(G,S,t)$ such that $S$ contains exactly three stacks $(c,h_1),(c,h_2),(c,h_3)$ of color $c$ in the given order and where $h_i<t$ for $i \in {1,2,3}$ while $h_1+h_2\geq t$.
    In any solution to $I$, the first and second stack need to be placed such that they are non-adjacent but have a common neighbor that must be empty when the third stack is placed.
\end{lemma}
\begin{proof}
    If the first two stacks are placed on adjacent vertices, they merge immediately and vanish, leaving the third stack unremoved at the end.
    If instead they are placed more than one vertex apart, they cannot later form a three-merge.
    Hence, exactly one vertex must separate them, onto which the third stack is placed to remove all of them in one go.
\end{proof}

Thus, two colors to which \Cref{lem:forced-three-merge} applies cannot alternate on a path as shown in~\Cref{fig:spider_gadget_cannot_three-merge_red_stacks}.

\begin{figure}[t]
    \begin{subfigure}[t]{.48\linewidth}
        \centering
        \includegraphics[page=2]{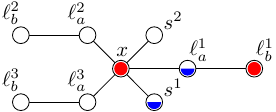}
        \caption{After placing $S_{initial}, S_{reservation}$ and $S$ the next red stack in $S_{connector}$ cannot three-merge the red stacks, violating \Cref{lem:constraint-non-overlapping}.}
        \label{fig:spider_gadget_cannot_three-merge_red_stacks}
    \end{subfigure}\hfill
    \begin{subfigure}[t]{.48\linewidth}
        \centering
        \includegraphics[page=4]{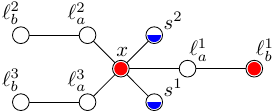}
        \caption{The admissible configuration of \Cref{lem:forced-weak-spider-config}.}
        \label{fig:spider_gadget_red_blue}
    \end{subfigure}
    \caption{Labeled graph $G'$ with red and blue stacks placed. The configuration depicted in (a) cannot lead to an empty graph, while (b) can.}
    \label{fig:spider_gadget}
\end{figure}

\begin{lemma}\label{lem:forced-four-merge}\label{lem:constraint-green-stack-independence}
    Let $c$ be a color of an \ehs instance $I=(G,S,t)$ such that $S$ contains exactly four stacks $(c,h_\ell),(c,h_\ell),(c,h_s),(c,h_s)$ of color $c$ in the given order and where $h_s<\frac{t}{2}<h_\ell$ such that $h_s+h_\ell\geq t$.
    In any solution to $I$, these stacks can vanish either (a) via two independent merges of an $h_\ell$-stack with an $h_s$-stack, or (b) by merging all four stacks in one move.
    In case~(a), the stacks are placed on the endpoints of two independent edges, in case (b) on a star of degree 3.
\end{lemma}
\begin{proof}
    Every color-$c$ stack is smaller than $t$, so no stack can vanish on its own. 
    Any merge involving exactly three color-$c$ stacks would cause precisely those three stacks to vanish, leaving a single stack that can never vanish.
    Merging the first two stacks (of height $2h_\ell>t$) would cause them to vanish, leaving the latter two stacks of total height $2h_s<t$ that cannot vanish.
    In contrast, two stacks of different height sum to at least $t$ and thus vanish.
    Thus, only the latter two $h_s$ stacks could be merged without vanishing, although this leaves no chance of merging them with the former $h_\ell$ ones to make them vanish.
    Altogether, this shows the color-$c$ stacks need to vanish in either (a) two pairs of two, or (b) all four in one go.
    In case (a), we first need to place the first two stacks without them merging, and then the third stack so that it only merges with one of them, necessitating the first edge.
    Finally placing the last stack adjacent to the leftover $h_\ell$ stack requires the second, independent edge.
    For case (b), the only way of merging all four stacks in one move is placing them on a star of degree 3, with the last stack placed on the center.
\end{proof}

Together, this allows us to show the following lemma; see also \Cref{fig:spider_gadget_red_blue}.

\begin{lemma}\label{lem:forced-weak-spider-config}
    After placing $S_{initial} + S_{reservation}$ in $I_H^* = (G', S', t)$, any solution needs to have the \cred stacks on \c and $\ell_b^1$, the \cblue stacks on $s^1,s^2$ and all other vertices empty.
\end{lemma}
\begin{proof}
    First, suppose that vertex \c remains empty after placing $S_{initial}$, i.e., the first two \cred and two \cblue stacks.
    \Cref{lem:constraint-vanishing-three-merge} now implies that four neighbors of \c must be occupied by the \cred and \cblue stacks, while all leaf vertices are empty.
    However, occupying four neighbors of \c leaves no space for the \cgreen stacks, inevitably violating \Cref{lem:constraint-green-stack-independence}.
    Hence, the central vertex \c must host a stack of color either \cred or \cblue, which means that we only need to consider case (a) of \Cref{lem:constraint-green-stack-independence}.

    After placing $S_{initial}$ vertex \c holds the stack $(c_a, t-1)$ where $c_a\in\{\cred,\cblue\}$.
    From \Cref{lem:constraint-vanishing-three-merge} it follows that the second $c_a$ stack of $S_{initial}$ must be placed at $\ell_b^i$ for some $i\in\{1,2,3\}$, say w.l.o.g.\ $\ell_b^1$.
    We now show that the two other stacks of $S_{initial}$, call their color $c_b$, must be placed on $s^1$ and $s^2$.
    First, if a $c_b$ stack is placed on any vertex of $\ell^2$ or $\ell^3$, there only remains one path of length two, which violates \Cref{lem:constraint-green-stack-independence}.
    Second, if a $c_b$ stack is placed on $\ell^1_a$, between the two existing $c_a$ stacks, this violates \Cref{lem:constraint-non-overlapping}, as $\ell^1_a$ must remain empty for the $c_a$ three-merge.
    Thus, only $s^1$ and $s^2$ remain.

    Finally, we show that $c_a=\cred$ and $c_b=\cblue$.
    Recall that \Cref{lem:constraint-non-overlapping} especially means that the common neighbor of the first two stacks needs to be empty when placing the third on it.
    As the third \cred stack appears before the third \cblue stack in $S_{connector}$, the common neighbor of the two previous \cred stacks needs to be empty at that point before placing $S_{connector}$.
    This is only the case when \c and $\ell_b^1$ are $\cred$ while $s^1$ and $s^2$ are $\cblue$.
    Removing the \cred stacks then frees up \c as common neighbor of the \cblue stacks.
\end{proof}

\begin{figure}[t]
    \centering
    \begin{subfigure}[b]{0.3\columnwidth}
        \centering
        \includegraphics[width=\textwidth,page=5]{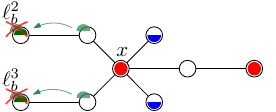}
        \caption{Step 1: Green stacks are placed, merged and vanish.}
        \label{fig:spider_gadget_progression:state1}
    \end{subfigure}
    \hfill
    \begin{subfigure}[b]{0.3\columnwidth}
        \centering
        \includegraphics[width=\textwidth,page=6]{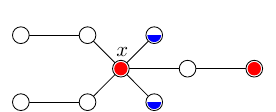}
        \caption{Step 2: Two empty independent paths remain.}
        \label{fig:spider_gadget_progression:state2}
    \end{subfigure}
    \hfill
    \begin{subfigure}[b]{0.3\columnwidth}
        \centering
        \includegraphics[width=\textwidth,page=7]{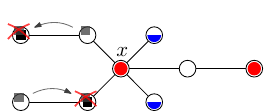}
        \caption{Step 3: \textsc{Partition} stacks are placed and merged.}
        \label{fig:spider_gadget_progression:state3}
    \end{subfigure}
    
    
    \begin{subfigure}[b]{0.3\columnwidth}
        \centering
        \includegraphics[width=\textwidth,page=8]{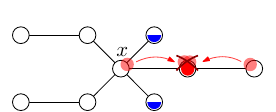}
        \caption{Step 4: Red stacks vanish.}
        \label{fig:spider_gadget_progression:state4}
    \end{subfigure}
    \hfill
    \begin{subfigure}[b]{0.3\columnwidth}
        \centering
        \includegraphics[width=\textwidth,page=9]{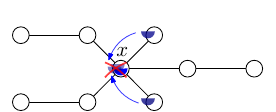}
        \caption{Step 5: Blue stacks vanish.}
        \label{fig:spider_gadget_progression:state5}
    \end{subfigure}
    \hfill
    \begin{subfigure}[b]{0.3\columnwidth}
        \centering
        \includegraphics[width=\textwidth,page=10]{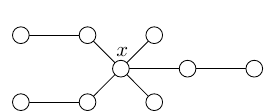}
        \caption{Step 6: $G'$ is empty.}
        \label{fig:spider_gadget_progression:state6}
    \end{subfigure}
    \caption{Spider graph stack placement progression leading to an empty graph.}
    \label{fig:spider_gadget_progression}
\end{figure}

\Cref{fig:spider_gadget_progression} visualizes how the two edges that are empty after \Cref{lem:forced-weak-spider-config} can be used to solve the partition instance encoded in $S$ and then clear the graph via $S_{connector}$.
\iftoggle{long}{
    Using \Cref{lem:forced-weak-spider-config}, we show that any solution must make these steps, thereby obtaining the following theorem.
}{
    In the full version \cite{todo}, we show that any solution must make these steps, thereby obtaining the following theorem.
}
\begin{restatable}\restateref{theorem:hexa-partition-gadget}{theorem}{thmNPPartitionCon}\label{theorem:hexa-partition-gadget}
    \ehs on spider graphs with nine vertices and four colors is weakly NP-hard.
\end{restatable}
\begin{prooflater}{proofNPPartitionCon}
    Let $P$ be an arbitrary instance of \textsc{Partition}. We construct the \ehs instance $I_H^* = (G', S', t)$ using the spider graph $G'$ and the extended stack sequence $S'$ as defined above.
    We show that the instance $P$ of Partition has a solution if and only if there is a solution to $I_H^*$ of \ehs.
    
    ($\Rightarrow$) Assume $P$ is a yes-instance with partition $P_1, P_2$ such that $\sum P_1 = \sum P_2 = t$.
    By \Cref{lem:forced-weak-spider-config}, the placement of $S_{initial}$ is fixed, occupying $x$, $\ell_b^1$, $s^1$, and $s^2$. 
    We place the \cgreen stacks of $S_{reservation}$ on the endpoints of the remaining edges $\ell^2$ and $\ell^3$, where they merge and vanish.
    This leaves exactly two disjoint edges available for the sequence $S$.
    We place the stacks corresponding to $P_1$ on $\ell^2$ and those of $P_2$ on $\ell^3$. Since $\sum P_1 = \sum P_2 = t$, both edges are cleared. 
    Finally, the sequence $S_{connector}$ clears the remaining gadget stacks by placing a \cred stack on $\ell^1_a$ and a \cblue stack on $x$. Thus, $I_H^*$ is a yes-instance.
    
    ($\Leftarrow$)
    Conversely, assume $I_H^*$ is a yes-instance. 
    The entire graph $G'$ must be emptied. 
    According to \Cref{lem:forced-weak-spider-config} established above, this is only possible if the placement of the \cblack stacks of $S$ successfully clears the two edges that remain after the \cgreen gadget stacks are placed and vanish. As shown by the argument in the proof of \Cref{thm:np-partition-discon}, clearing $\ell^2$ and $\ell^3$ requires the set $S$ to be partitioned into two subsets $S_1$ and $S_2$, placed onto $\ell^2$ and $\ell^3$ respectively, such that $\sum S_1 = \sum S_2 = t$. This construction immediately yields a valid partition $P_1, P_2$ for $P$, thus proving that $P$ is a yes-instance.
    
    As the prior reduction is polynomial and the gadget adds only five vertices, six edges, and ten stacks at most as large as $t$, this reduction from \textsc{Partition} to $I^*_H$ can also be done in polynomial time.
\end{prooflater}

\subsection{Reduction from 3-Partition}\label{sec:bin-gadget}\label{sec:strong-np}
Note that by reducing from \textsc{Partition} in the last section, we rely on large numbers and therefore can only show weak NP-hardness.
In this section, we now want to show strong NP-hardness, that is hardness even when $t$ (and thereby also all stack heights) is bounded by a polynomial in the input size \cite{GareyJohnson1978StrongNP}.
For this, we use a reduction from the strongly NP-hard problem \textsc{3-Partition}, where we are given a multiset $A$ of $3m$ positive integers and a bound $B \in Z^+$.
The question is whether $A$ can be partitioned into $m$ triplets $A_1, A_2, \dots, A_m$ such that each triplet sums up to exactly $B$.
As noted by Garey et al.~\cite[p.~224]{GareyJohnson1979} the problem is NP-complete in the strong sense, and we can assume that every input value $a \in A$ is restricted such that $B/4 < a < B/2$ and $\sum_{a \in A} a = mB$.

For a reduction from \textsc{3-Partition} to \ehs, we will adapt our prior gadget to now leave exactly three distinct, isolated vertices available for the placement of the three elements of a single triplet.
Using $m$ disjoint copies of this gadget, we can thereby model the $m$ triplets in a \textsc{3-Partition} solution.
More specifically, our reduction works as follows.
A~gadget consists of one central vertex \c of degree five, four arms of length two with vertices $\ell_a^i,l_b^i$ for $i \in \{1,2,3,4\}$, and one arm of length one, $s^1$; see \Cref{fig:3part_spider_gadget}.

\begin{figure}[t]
    \centering
    \begin{subfigure}[b]{0.3\columnwidth}
        \centering
        \includegraphics[width=\textwidth,page=2]{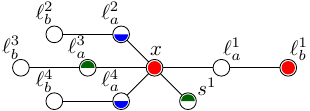}
        \caption{The admissible configuration of \Cref{lem:forced-strong-spider-config}.}
        \label{fig:3part_spider_gadget_lemma}
    \end{subfigure}
    \hfill
    \begin{subfigure}[b]{0.3\columnwidth}
        \centering
        \includegraphics[width=\textwidth,page=3]{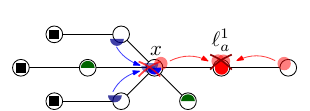}
        \caption{Removing the gadget stacks \inlineRedStack, \inlineBlueStack and finally \inlineGreenStack.}
        \label{fig:3part_spider_gadget_after}
    \end{subfigure}
    \hfill
    \begin{subfigure}[b]{0.3\columnwidth}
        \centering
        \includegraphics[width=\textwidth,page=7]{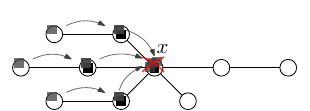}
        \caption{Pulling the \inlineBlackStack stacks to merge and vanish on \c.}
        \label{fig:3part_spider_pull}
    \end{subfigure}
    \caption{The spider gadget used in our \textsc{3-Partition} reduction.}
    \label{fig:3part_spider_gadget}
\end{figure}

Given a \textsc{3-Partition} instance $(A,B)$ we obtain an \ehs instance $I_H=(G,S,t)$ where graph $G$ has $m=\frac{|A|}{3}$ disjoint copies of the gadget graph and the threshold is $t=B+4$.
For $S$, we use $3m+1$ colors: \cblack for the triplet elements, and three gadget colors $\cred^j$, $\cblue^j$, $\cgreen^j$ for $j\in\{1,\ldots,m\}$.
The input sequence is constructed as follows.
\begin{itemize}
    \item The initial gadget stack sequence $S_{initial}^j = \langle (\cred^j, t-1), (\cred^j, t-1), (\cblue^j, t-1),$\\$(\cblue^j, t-1), (\cgreen^j, t-1), (\cgreen^j, t-1) \rangle$ forces the isolation of three vertices.
    \item The sequence $S_{A}=\langle(\cblack,a_1),\ldots,(\cblack,a_{3m})\rangle$ represents $A=\{a_1,\ldots,a_{3m}\}$.
    \item The connector gadget stack sequence $S_{connector}^j = \langle (\cred^j, t-1), (\cblue^j, t-1), (\cgreen^j,t-1) \rangle$ frees each gadget of the (initial) gadget stacks.
    \item The triplet connector sequence
    $S_{final}^j = \langle (c_{black},1),(c_{black},1),(c_{black},1),(c_{black},1)\rangle$ draws the placed triplet, causing all its stacks to merge and vanish.
\end{itemize}
We construct the full input sequence as $S = S_{initial} + S_{A} + S_{connector} + S_{final}$, where $S_{initial}$, $S_{connector}$ and $S_{final}$ are obtained by concatenating all $S_{initial}^j$, $S_{connector}^j$, and $S_{final}^j$, respectively, for all $j\in\{1,\ldots,m\}$.
Observe that \Cref{lem:forced-three-merge} still applies to all $\cblue^j$, $\cred^j$, and $\cgreen^j$.
Analogously to \Cref{lem:forced-weak-spider-config}, we again show that $S_{initial}$ forces a certain configuration; see \Cref{fig:3part_spider_gadget_lemma}.

\begin{restatable}\restateref{lem:forced-strong-spider-config}{lemma}{lemForcedStrongSpiderConfig}\label{lem:forced-strong-spider-config}
    After placing $S_{initial}$ in $I_H$, any solution has, for each gadget,
    the vertices $\ell_a^i$ and all $\ell_b^{j \neq i}$ empty for some $i \in \{1,2,3,4\}$.
    Furthermore, the two neighbors of $\ell_a^i$ have the same color $c$, and each gadget consists of exactly three colors, where $c$ appears first in $S_{final}$.
\end{restatable}
\begin{proofsketch}
    First observe that no color can be split onto multiple disconnected gadgets due to \Cref{lem:forced-three-merge}.
    Using a case analysis together with the pigeonhole principle, one can then show that every gadget contains exactly six stacks, which then need to use exactly three colors.
    The positions of empty vertices then follow similar to \Cref{lem:forced-weak-spider-config}.
\end{proofsketch}
\begin{prooflater}{proofForcedStrongSpiderConfig}
    First, we analyze the maximum number of stacks a gadget can have. 
    Suppose that a gadget contains at least seven stacks.
    Since a gadget has ten vertices, at least one stack must be placed on a leaf $\ell_b^i$ for some $i \in \{1,2,3,4\}$. Then \Cref{lem:forced-three-merge} forces the center \c to hold a stack of the same color and vertex $\ell_a^i$ to remain empty. Since \c is occupied, vertices $\ell_b^{j \neq i}$ must remain empty. This leaves exactly six usable vertices ($x$, $\ell_b^i$, $s^1$, and three $\ell_a^{j \neq i}$), contradicting our assumption that a gadget can hold seven stacks.
    It follows that each gadget can contain at most six stacks. Since $S_{initial}$ consists of $6m$ stacks and there are $m$ gadgets, every gadget must contain exactly six stacks; otherwise, by the pigeonhole principle, some gadget would contain at least seven stacks, which we have shown to be impossible.

    Furthermore, as \Cref{lem:forced-three-merge} prevents splitting colors across gadgets, the six stacks must consist of exactly three colors.

    Next, we analyze the stack placement in each gadget. Suppose that the central vertex \c is empty after placing six stacks in a gadget. By \Cref{lem:forced-three-merge}, each pair of colors in $S_{initial}$ must be placed exactly one vertex apart, forcing all five neighbors of \c to be occupied. Therefore, the sixth stack cannot be placed without violating the lemma. 
    Hence, \c must contain a stack. By \Cref{lem:forced-three-merge} one leaf $\ell_b^i$ must contain a stack of the same color and the specified configuration above follows.
    
    Lastly, we analyze the colors in each gadget. We say that a gadget has colors $c$, $c'$, and $c''$, where $c$ is on $x$ and $\ell_b^i$.
    Observe that the permutation of $c'$ and $c''$ on $\ell_a^{i\neq j}$ and $s^1$ is irrelevant. 
    As stated above, by \Cref{lem:forced-three-merge} the vertex $\ell_a^i$ must remain empty in order to enable a three-merge. Therefore the exact placement of the gadget colors within $S_{initial}$ and $S_{final}$ does not matter, as long as $c$ appears in $S_{final}$ before $c'$ and $c''$.
\end{prooflater}

We will again assume without loss of generality that $\ell_a^1$ as well as $\ell_b^2,l_b^3,l_b^4$ are always the empty vertices.
As each gadget received exactly three colors, where each color needs to be put onto a single gadget, the exact position of a color within $S_{initial}$ and $S_{final}$ is irrelevant, as long as $c$ appears first in $S_{final}$. We assume w.l.o.g.\ that the $j$-th gadget received colors $\cred^j$, $\cblue^j$, $\cgreen^j$ such that its \c is $\cred^j$, while $s^1$ and $\ell_a^2$ are $\cgreen^j$; see \Cref{fig:3part_spider_gadget_lemma}.
We can now show the following.

\begin{theorem}\label{theorem:3partition-hexasort-strongly-np}
    \ehs on the disjoint union of spider graphs with ten vertices each is strongly NP-hard.
\end{theorem}
\begin{proof}
    From any solution $(A,B)$ we construct a valid placement strategy for $I_H$ as follows.
    Place $S_{initial}$ such that each gadget is in the configuration depicted in \Cref{fig:3part_spider_gadget_lemma}, leaving $\ell^2_b, l^3_b,$ and $\ell^4_b$ available.
    For each gadget $j$, we place the three elements of the triplet $A_j$ on these three vertices.
    Next, we process $S_{connector}$ and place the $\cred^j$ stack on $\ell^1_a$, and the $\cblue^j$ and $\cgreen^j$ connector stacks on \c.
    This leaves the gadget empty except for the stacks from $S_A$ on the leaves.
    Finally, we distribute the sequence $S_{final}$ evenly: for each gadget, we place three \cblack stacks on $\ell^2_a, l^3_a, l^4_a$ and the fourth on \c.
    This pulls the triplet stacks onto \c. Since the sum of the triplet is exactly $B$ and the four black stacks contribute $4$, the total height on \c becomes $B+4 = t$ and the stack vanishes.
    Since we apply this strategy to all $m$ gadgets, the entire graph is emptied, and $I_H$ is a YES-instance.
    
    For the converse direction, we continue, where \Cref{lem:forced-strong-spider-config} left off.
    Here, as $\ell_a^1$ cannot be used by $S_A$ due to \Cref{lem:constraint-non-overlapping}, exactly the three vertices $\ell^2_b, l^3_b,$ and $\ell^4_b$ remain available for placing $S_A$.
    It is easy to see that, again due to pigeonhole principle, exactly three stacks from $S_A$ must be placed on them.
    Next, the stacks of $S_{connector}$ are placed.
    Since for each gadget $j$ now only $\ell^1_a$ is empty, the corresponding $\cred^j$ connector stack is forced onto this vertex, merging the red triplet.
    The $\cblue^j$ and $\cgreen^j$ connector stacks are then forced to three-merge on \c, see \Cref{fig:3part_spider_gadget_after}.
    This leaves the gadget empty except for $\ell^2_b, l^3_b,$ and $\ell^4_b$ containing stacks from $S_A$.
    Emptying the gadget then requires at least four \cblack height-1 stacks from $S_{final}$ to pull the triplet stacks onto \c; see \Cref{fig:3part_spider_pull}.
    By the pigeon hole principle, each gadget can receive at most four such stacks, contributing a value of 4 to the required height of $t=B+4$.
    Thus, each gadget can only be emptied if its three stacks from $S_A$ sum to at least $B$. Since the total height of $S_A$ distributed across all $m$ gadgets is exactly $mB$, each triplet must sum to exactly $B$, implying a \textsc{3-Partition} solution to $(A,B)$.

    Our reduction can be performed in polynomial time and maintains the numerical values, since it constructs $m$ constant-size gadgets and a sequence where the total number of stacks and colors are in $\mathcal{O}(m)$, while the threshold and all stack heights are bounded by $\mathcal{O}(B)$, thereby establishing strong NP-hardness.
\end{proof}

We again want to extend our result to connected graphs.
For this, we will connect the~$s^1$ vertices of all gadgets in a tree-like manner:
Let $T$ be a tree on $m$ vertices $t_1,\ldots,t_m$.
For the \ehs instance $I_H^*=(G',S,t)$, the connected graph $G'$ is obtained by combining the $m$ gadgets of $G$ from above with $T$: for all $j\in\{1,\ldots,m\}$, identify the $s^1$ vertex of the $j$-th gadget in $G$ with the tree node $t_j$ in $T$.
Observe that, if $T$ has bounded height or degree, so does $G'$.
It remains to show that \Cref{lem:forced-strong-spider-config} still applies, which in this case hinges on stacks of the same color not being distributed over multiple gadgets.

\begin{lemma}\label{lem:strong-connected-config}
    After placing $S_{initial}$ in $I_H^*$, any solution needs to have all stacks of exactly three colors on each gadget.
\end{lemma}
\begin{proof}
    We first analyze the maximum number of stacks a gadget can have. 
    Suppose that a gadget contains at least seven stacks.
    Since a gadget has ten vertices, at least one stack must be placed on a leaf $\ell_b^i$ for some $i \in \{1,2,3,4\}$. 
    Then \Cref{lem:forced-three-merge} forces the center \c to hold a stack of the same color.
    Specifically, in the way $G'$ is constructed, the same lemma forces vertex $\ell_a^i$ and then also all $\ell_b^{j \neq i}$ to remain empty.
    Thus, no gadget can have seven or more stacks, which via the pigeon hole principle implies that all gadgets have exactly six stacks.

    However, since the gadgets are connected, they could contain more than three colors, with colors split across multiple gadgets.
    Note that, due to \Cref{lem:forced-three-merge}, the only vertices that can host a (non-black) stack that is split across multiple gadgets are \c and $s^1$, see \Cref{fig:3partition_tree-stack-placement-across-gadgets}.

    \begin{claim}\label{claim:four-colors-gadget}
    If a gadget hosts a color that is split across multiple gadgets, then $\ell^i_b$ needs to be empty for all $i\in\{1,2,3,4\}$, and both \c and $s^1$ must contain split colors.
    \end{claim}
    \begin{claimproof}
        Suppose that a gadget contains a split color on $x$.
        Then, by \Cref{lem:forced-three-merge}, no stack is placed on any $\ell_b^i$ vertex for all $i\in\{1,2,3,4\}$.
        Hence, to satisfy the requirement that each gadget contains six stacks, in addition to $x$; $s^1$ and all $\ell_a^i$ for all $i\in\{1,2,3,4\}$ must be occupied.       
        Since $S_{initial}$ contains only pairs of colored stacks, at least three colors must be distributed across these five vertices; therefore, at least one of these vertices contains a color that appears only once on the gadget.
        By \Cref{lem:forced-three-merge}, no $\ell_a^i$ for any $i\in\{1,2,3,4\}$ can contain a color that appears only once.
        Therefore, two pairs of colored stacks must be split across the $\ell_a^i$ vertices, and a stack of a fourth color must be placed on $s^1$.
        The same logic applies if we assume that the split color is placed on $s^1$.
    \end{claimproof}

    From this, we can conclude that, if either of \c and $s^1$ hosts a color that only appears once on the respective gadget, so does the other with a different color that only appears once on the gadget.

    We will now take an arbitrary such pair, remove it from the graph, and show that the smaller remainder graph must always contain another such pair, which contradicts $G'$ being a finite acyclic graph.
    Consider an \c vertex of gadget $j$ that has a color $c_1$ that only appears once on the gadget.
    Due to \Cref{lem:forced-three-merge}, the second stack of $c_1$ needs to be on the $s^1$ vertex of gadget $j'$ such that $t_j$ and $t_{j'}$ are adjacent in $T$; see \Cref{fig:3partition_tree-stack-placement-across-gadgets}.
    As noted above in the proof of \Cref{claim:four-colors-gadget}, vertex \c of gadget $j'$ needs to have a different color $c_1'$ that only appears once on that gadget.
    In particular, the $s^1$ vertex of gadget $j$ cannot have color $c_1'$ as otherwise it would alternate with $c_1$, violating \Cref{lem:forced-three-merge}.
    We remove gadget~$j$, and continue by considering the pair of color $c_1'$ on the smaller, connected graph that is the connected component of the remaining gadget $j'$.
    This yields a contradiction, as the argument can continue indefinitely while the graph is finite.
\end{proof}

\begin{figure}[t]
    \centering
    \includegraphics[width=\textwidth,page=3]{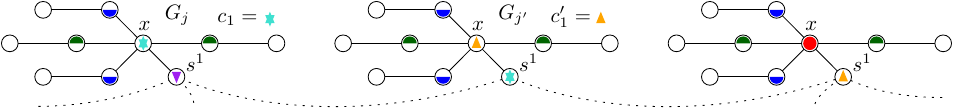}
    \caption{A connected graph $G'$ with stacks placed across multiple gadgets excluded by \Cref{lem:strong-connected-config}.}
    \label{fig:3partition_tree-stack-placement-across-gadgets}
\end{figure}

\iftoggle{long}{
    The remainder of \Cref{lem:forced-strong-spider-config} and thereby the following theorem now follows.
}{
    The remainder of \Cref{lem:forced-strong-spider-config} and thereby the following theorem now follows; see \cite{TODO}.
}

\begin{theorem}\label{theorem:3partition-hexasort-strongly-np-tree}
    \ehs on trees of bounded height or degree is strongly NP-hard.
\end{theorem}
\onlyLong{
    We show that despite the added connectivity, the gadgets must be solved independently.
    After placing $S_{initial}$, $S_A$ and $S_{Connector}$ emptying each gadget still requires at least four \cblack height-1 stacks.
    Merging triplets of different gadgets necessarily requires more than four \cblack height-1 stacks from $S_{final}$ per gadget, because the triplet on central vertex $x$ would have to be pulled over to $s^1$, requiring another \cblack height-1 stack.
    However, by the pigeonhole principle, another gadget would then have less than four \cblack height-1 stacks available , and at least one stack would remain.
    Since this effectively rules out any interaction between gadgets, the correctness of the reduction follows directly from \Cref{theorem:3partition-hexasort-strongly-np}, thereby locally enforcing a valid \textsc{3-Partition} solution on each gadget.
    Furthermore, the reduction preserves the polynomial bounds established in the theorem.
}


\section{Algorithms}\label{sec:dynamic-programming}
In this section we investigate algorithmic approaches towards solving \textsc{Hexasort}.
Note that a trivial brute-force approach to decide an instance of \textsc{Fitting} or \ehs runs in $O(|V|^{|S|})$ time:
We try for each stack in the sequence $S$ each vertex on which it could be placed and then check whether the final configuration is admissible.
We develop a dynamic programming approach to \textsc{Hexasort} by instead systematically exploring the space of reachable board configurations, yielding running time exponential in the graph~size.

\begin{theorem}\label{thm:dp-states}
    \textsc{Fitting} and \ehs can be decided in time $\mathcal{O}(|S| \cdot (|C|\cdot t)^{|V|})$.
\end{theorem}
\begin{proof}
    The core of our approach is a Boolean table, whose entry $\mathrm{DP}[i, \delta]$ is true if and only if the board configuration $\delta$ is attainable after placing the first $i$ stacks of $S$. 
    We initialize the dynamic program by setting only $\mathrm{DP}[0, \delta_0] = \text{true}$ for the unique empty configuration $\delta_0$ where $\delta_0(v) = \perp$ for all $v \in V$. 
    The recurrence relation processes the sequence step-by-step: for each index $i$ from $1$ to $m$ and every configuration $\delta$ such that $\mathrm{DP}[i-1, \delta] = \text{true}$, we consider every possible valid move.
    A move consists of placing the current stack $s_i = (c, h)$ on any \emph{empty} vertex and sets $\mathrm{DP}[i, \delta'] = \text{true}$ for the resulting configuration $\delta'$.
    An \ehs instance is positive if and only if $\mathrm{DP}[m, \delta_0] = \text{true}$ upon exhaustion of the sequence $S$.
    A \fhs instance is positive if $\mathrm{DP}[m, \delta_0] = \text{true}$ holds for any $\delta$.

    Regarding the running time, observe that there are $\mathcal O((|C|\cdot t)^{|V|})$ configurations as each of the $|V|$ vertices has $|C| \cdot (t-1) + 1$ possible states.
    Since the algorithm iterates through the sequence and considers up to $|V|$ placement options for each reachable state, the total time complexity is $\mathcal{O}(|S| \cdot (|C| \cdot (t-1) + 1)^{|V|} \cdot |V|) = \mathcal{O}(|S| \cdot (|C|\cdot t)^{|V|})$. 
\end{proof}

While this algorithm becomes slower the larger the graph becomes, we will now show that once the graph becomes large enough to provide dedicated space for each color, the problem becomes trivially positive.
For this, we need separate approaches for both problem variants.

\begin{lemma}\label{thm:fit-greedy}
    Any \fhs instance where $G$ contains a matching of size at least $|C|$ is positive.
\end{lemma}
\begin{proof}
    For each color $c \in C$, we reserve an edge from the matching.
    When a stack $(c,h)$ arrives, it is placed alternatingly on one of the two endpoint vertices.
    Thereby, for each color at any point in time, at least one free vertex is available.
    Thus, all stacks can be placed.
\end{proof}

\begin{lemma}\label{thm:fit-degree}
    Any \fhs instance where $G$ contains a vertex of degree $|C|\cdot t$ is positive.
\end{lemma}
\begin{proof}
    If $G$ contains a vertex \c of degree at least $|C|\cdot t$, we can place the stacks from $|S|$ on any of its neighbors, only placing a stack directly on \c whenever its color has already $t$ stacks adjacent, thereby eliminating all stacks of this color.
    As no color ever has more than $t$ stacks on the graph and any stack placed on \c immediately vanishes (the merge of its $t$ neighbors surely exceeds size $t$), all stacks can be placed and the instance is thus always positive.
\end{proof}

\begin{observation}\label{thm:empty-negative-greedy}
    Any \ehs instance where the stacks of some color $c$ sum up to less than $t$ when ignoring all height-$t$ stacks is negative if the last color-$c$ stack has height less than $t$.
\end{observation}

\begin{lemma}\label{thm:empty-greedy}
    Any \ehs instance, that is not trivially negative by \Cref{thm:empty-negative-greedy}, and where $G$ contains $|C|$ non-overlapping spiders, each with two length-two legs and one length-one leg, is positive.
\end{lemma}
\begin{proof}
    For each color $c \in C$, we reserve a degree-3 spider with center \c adjacent to the first vertices of the three legs $s$, $\ell_a^1\ell_b^1$, $\ell_a^2\ell_b^2$.
    Let $z=(c,h_z)$ be the last color-$c$ stack in $S$.
    If $h_z=t$ or $sum(c, S)-h_z<t$, place all color-$c$ stacks alternatingly on any two adjacent vertices of the spider.
    Otherwise, let $X$ be an arbitrary minimal set of color-$c$ stacks, excluding any height-$t$ stack, such that $\Sigma X\geq t$.
    Note that by assuming that \Cref{thm:empty-negative-greedy} does not make the instance trivially negative, such set must exist.
    Now let $y=(c,h_y)$ be the first stack of $X$ in $S$, which we place on $s$.
    Place all other stacks of $X$, except for $z$, alternatingly on $\ell_a^1$ and $\ell_b^1$ such that the last stack of $X$ is placed on $\ell_a^1$. 
    Place all other stacks of color $c$, except for $z$, alternatingly on $\ell_a^2$ and $\ell_b^2$ such that the last of them is on $\ell_a^2$.
    Place the last stack $z$ on the center \c.
    By minimality of $X$, whether or not $z \in X$, at this point $s$ contains $y$ and $\ell_a^1$ a stack that reaches height at least $t$ when merged with $y$ (or with $z$ and $y$).
    When $z$ is placed on \c, it merges $\ell_a^1$ and $s$ with any remaining stack on $\ell_a^2$, which makes color $c$ vanish.
    Consequently, all stacks in $S$ can be placed while emptying the graph in the end.
\end{proof}

For \fhs, combining these results now allows us to show fixed-parameter tractability when parametrizing by the number of colors and the threshold value.

\begin{theorem}\label{thm:fpt}
    \fhs can be decided in time $\mathcal{O}(f(|C|, t)\cdot |S|^2)$ for some function $f$.
\end{theorem}
\begin{proof}
    First note that if $G$ is small enough such that $|V|$ is bounded by $|C|$, the instance can be solved in $\mathcal{O}((|C|\cdot t)^{|C|} \cdot |S|)$ by \Cref{thm:dp-states}.
    If $G$ is instead large and contains a matching of size at least $|C|$, the instance is always positive by \Cref{thm:fit-greedy}.
    The existence of such matching can be checked in polynomial time; see \cite{10.1145/2529989} for a survey of algorithms.
    Also, by \Cref{thm:fit-degree}, all instances with a vertex of degree at least $|C|\cdot t$ are always positive.
    It remains to treat the case where $G$ is large but has bounded degree and does not contain a large enough matching.
    Equivalently, the latter means that $G$ has a vertex cover $X$ of a size bounded by $2\cdot |C|$ \cite{sasak2010comparing,jacob2025treebasedvariantbandwidthforbidding}.

    Let $V^0$ denote the set of isolated vertices. 
    It is easy to see that the instance is always positive if $|V^0| \geq |S|$ -- place every stack on a single vertex.
    Thus $G$ has a bounded vertex cover $X$, no high-degree vertex, and a bounded number of isolated vertices $V^0$.
    This implies that, disregarding the isolated vertices $V^0$, graph $G$ has a size bounded by $|C|^2\cdot t$.
    This is because each of the at most $2\cdot |C|$ vertices of $X$ can have at most $|C|\cdot t$ neighbors outside $X$, and any non-isolated vertex needs to either be in $X$ or have a neighbor in $X$.
    The running time of the dynamic program from \Cref{thm:dp-states} can be factored as $\mathcal{O}(|S| \cdot (|C|\cdot t)^{|V\setminus V^0|} \cdot (|C|\cdot t)^{|V^0|})$.
    As isolated vertices are structurally identical and cannot affect each other, we do not need to track their specific configurations.
    Instead, we simply count the number of non-empty isolated vertices and thus simplify the exponential factor $((|C|\cdot t)^{|V^0|})$ with the linear factor $|V^0|$.
    This yields a running time of $\mathcal{O}(|S| \cdot (|C|\cdot t)^{|V\setminus V^0|} \cdot |V^0|)$. Since $|V \setminus V^0|$ is bounded by $|C|^2\cdot t$ and $|V^0|$ is bounded by $|S|$ the complexity simplifies to $\mathcal{O}((|C|\cdot t)^{|C|^2\cdot t} \cdot |S|^2)$ satisfying our bound in this final case.
\end{proof}

\section{Conclusion and Future Work}\label{sec:conclusion}
In this work, we studied the computational complexity of the popular game \textsc{Hexasort} in the \textsc{Empty} and \textsc{Fitting} variants.
In \Cref{sec:np-hardness} we presented multiple NP-hardness reductions, starting with a reduction from the \textsc{Partition} problem to \ehs with one color on two independent paths of size two.
Reducing \textsc{Empty} to \fhs by appending four stacks with new colors to such instance, this result also applies to the latter variant when allowing five colors.
We further used \textsc{3-Partition} to show that \ehs is strongly NP-hard, even when restricted to trees of bounded height or degree.
This demonstrates that the problem remains intractable even under severe structural restrictions on the graph, together with either a constant number of colors or numbers that are polynomially-bounded by the input size.

At the same time in \Cref{sec:dynamic-programming}, we identified several classes of instances that admit efficient algorithms.
We introduced a dynamic programming algorithm that solves both variants on arbitrary graphs in time $\mathcal{O}\bigl(|S| \cdot (|C|\cdot t)^{|V|}\bigr)$.
Note that this immediately implies tractability on constant-size graphs when the threshold $t$ is polynomially-bounded by the input size.
We also showed that instances with sufficiently large graphs, which allow each color to be solved independently on its own exclusive part of the graph, are always tractable.
Based on this, we were able to show that \fhs is fixed-parameter tractable when parametrizing by the number of colors and the threshold value.

For \fhs we thus obtain results for bounding any two of the graph size $|V|$, the number of colors $|C|$ and the threshold $t$.
While the problem is (weakly) NP-hard for constant $|V|$ and $|C|$, it becomes tractable for either constant $|V|$ and polynomially-bounded $t$ or for constant $|C|$ and $t$.
We leave the complexity of the case where $t$ is bounded by a constant as open question.
Furthermore, it would be nice to extend our results to show also strong NP-hardness for constant $|C|$.
%
Similarly, our results for \ehs are tight in the sense that, when bounding both the graph size by a constant as well as the elimination threshold by a polynomial in the input size, we obtain a polynomial-time algorithm.
In contrast, only applying either of the bounds yields NP-hardness.
We leave as an open question whether \ehs becomes tractable when bounding both the threshold and the number of colors, e.g. by using \Cref{thm:empty-greedy} as we did for the \textsc{Fitting} case.
The main challenge here is arguing that any graph containing only a few copies of this more involved structure needs to have a simple structure itself.

Regarding the \fhs instances found in playable versions of the game, their hexagonal grid is usually large enough w.r.t.\ the number of colors to allow for a sufficiently large matching.
What instead complicates these instances is that stacks here are also allowed to contain hexagons of multiple different colors, which is a further direction in which our results could be extended.



\bibliography{references}

\appendix

\end{document}